\documentclass[letterpaper, 10 pt, conference]{article}

\newcommand{\setlrmargins}[2]{
  \hoffset -1in
  \textwidth              \paperwidth
  \advance\textwidth      -#1             
  \advance\textwidth      -#2             
  \oddsidemargin  #1
  \evensidemargin #2
  }

\setlrmargins{1in}{1in}
\renewcommand{\baselinestretch}{1.1}

\usepackage{subcaption}
\usepackage{graphics} 
\usepackage{epsfig} 
\usepackage{mathptmx} 
\usepackage{times} 
\usepackage{amsmath} 
\usepackage{amssymb}  
\usepackage{float}
\usepackage{bbm}

\usepackage{amsthm}

\newtheorem{lemma}{Lemma}

\newtheorem{assumption}{Assumption}
\newtheorem{proposition}{Proposition}

\def\QEDclosed{\mbox{\rule[0pt]{1.3ex}{1.3ex}}} 
\def\QED{\QEDclosed} 

\title{\LARGE \bf On the Role of Asymptomatic Carriers in Epidemic Spread Processes}

\author{Xiaoqi Bi and Carolyn L. Beck 
\thanks{Xiaoqi Bi and Carolyn L. Beck are with the Coordinated Science Laboratory and the ISE Department, University of Illinois at Urbana-Champaign, Urbana, IL 61801, Emails: $\{\text{xiaoqib2} | \text{beck3} \text{}\}$@illinois.edu
      }%
}

\begin{document}

\maketitle

\begin{abstract}
We present an epidemiological compartment model, SAIR(S), that explicitly captures the dynamics of asymptomatic infected individuals in an epidemic spread process. We first present a group model and then discuss networked versions. We provide an investigation of equilibria and stability properties for these models, and present simulation results illustrating the effects of asymptomatic-infected individuals on the spread of the disease. We also discuss local isolation effects on the epidemic dynamics in terms of the networked models. Finally, we provide initial parameter estimation results based on simple least-squares approaches and local test-site data.\par
Keywords: Epidemic dynamics, networks, data-informed modeling, stability analysis, parameter estimation
\end{abstract}
\vspace*{-.05in}
\section{Introduction}
Modeling, analysis and control of epidemic spread processes over networks have received increasing attention over the past decade, owing not only to the recent COVID-19 pandemic, but also to recent outbreaks of the related SARS and MERS viruses, the Zika and Ebola viruses, and as well the plethora of computer network viruses. Conducting experiments to analyze infectious disease spread processes and response policies are prohibitive for many reasons, and effectively impossible over large human contact networks. As a result, mathematical modeling and simulation, informed by up-to-date data, provides an essential alternative for estimating and predicting when and how an epidemic will spread over a network. Epidemic models can be used to predict and plan for viral extent, in particular after validating the models with data collected during actual outbreaks. Moreover, simulations of strategic control policies for validated epidemic models can provide insights into approaches for mitigating virus spread over networks.

Mathematical models for epidemics, or spread processes, have been proposed, analyzed and studied for over $200$ years \cite{Bernoulli}.  The base models for most studies today derive from the so-called {\em compartment models} proposed by Kermack and McKendrick in $1932$ \cite{kermack1932contributions}.  These models assume every subject lies in some segment or compartment of the population at any given time, with these compartments possibly including {\em susceptible} (S), {\em infected} (I), {\em exposed} (E) and/or {\em recovered} (R) population groups, leading to the classical epidemiological models: SI (susceptible-infected), SIS (susceptible-infected-susceptible), SIR (susceptible-infected-recovered) and SEIR (susceptible-exposed-infected-recovered) models. As one example, the Kermack and McKendrick SIS model is given by 
\begin{equation}\label{eq:SIS}
\begin{array}{rcl} 
\dot{S}(t) &=& -\beta S(t)I(t) + \delta I(t) \\
\dot{I}(t) &=& \beta S(t) I(t) - \delta I(t), \end{array}
\end{equation}
where $S(t)$ is the susceptible segment of the population, $I(t)$ is the infected segment of the population, $\beta$ represents the rate of infection or contact amongst infected and susceptible subgroups, and $\delta$ represents the healing or curing rate. This foundational model assumes: (1) a homogeneous population with no vital dynamics, that is birth and death processes are not included, meaning that infection and healing are assumed to occur at faster rates than vital dynamics and the population size is assumed to remain constant; and (2) the population mixes over a trivial network, or in other words, over a complete graph structure. These assumptions have led to errors in previous epidemic forecasts \cite{Scarpino19}. 

We note that similar models to that given in (\ref{eq:SIS}) have been derived for SI, SIR(S) and SEIR(S) processes; SI models simply have $\delta = 0$; SIR(S) models include a recovered segment of the population and a recovery rate $\gamma$; and SEIR(S) models include an exposed segment of the population and a corresponding parameter $\sigma$ capturing the rate at which an exposed individual transitions to the infected state; the exposed segment is typically assumed to be non-infectious with the accompanying rate parameter capturing the disease incubation period.
There are numerous variants of these models, including recent models in which human awareness is taken into account
\cite{Funk2009},\cite{Arenas_2013},\cite{paarporn2017networked},\cite{ji2017bivirusAware}, and in which multiple epidemic processes may be propagating simultaneously \cite{xu2012multi},\cite{acc_multi},\cite{PLBNB:Auto20}. 

Over the past two decades, both to address the discrepancies found in prior epidemic forecasts, and to better model spreading processes of computer viruses over communication networks, there has been an extensive study of epidemic processes evolving over complex network structures; see for example \cite{KephartWhite91}, \cite{Towsley05}, \cite{DraeifMassoulie10}, \cite{pastor2015epidemic}, and from a controls perspective \cite{NowzariPP16}.\footnote{The literature in this area is vast, thus we cannot provide an overview of all prior research due to space constraints. However, we note that the cited papers provide extensive summaries of existing results.} 
To account for network structure among members of a population, an agent-based perspective of epidemic processes is taken where each agent is represented by a node in the network, and the edges in the network between nodes represent the strength of the interaction between agents. Nodes in the network may represent either individuals or subgroups in the larger population.
Given a total of $n$ such nodes, epidemic processes can be described by large Markov process models (e.g., of dimension $2^n$ for SIS models and $3^n$ for SIR models), which capture the probability of each node transitioning from susceptible to infected, and/or to recovered states, and back. These probabilities are determined by the infection rate(s), healing rate(s) and/or recovery rate(s), in addition to the network interconnection structure, and capture the stochastic evolution of such epidemic processes. As $n$ increases, these models quickly become intractable to analyze due to their size, 
at which point it is assumed that {\em mean-field approximation} (MFA) models are appropriate; these models are derived by taking expectations over infection transition rates of the agents and rely on the fundamental work of Feller \cite{Feller40} and Kurtz \cite{Kurtz81}.

For agents interconnected via a graph with adjacency matrix $W = [W_{ij}]$,
where element $W_{ij}$ defines the strength of the connection from node $i$ to node $j$, using the assumptions of large and constant agent population size along with additional independence assumptions, the deterministic networked MFA dynamic models are now widely applied; these models have been analyzed in detail and
shown to provide upper bounds on the probability of infection of a given agent at any given time (see \cite{VanMieghem09} and \cite{ChatterjeeD09} for discussions and perspectives). 
Again considering an SIS process example, denoting the probability of node $i$ being infected at time $t$ by $p_i(t) \in [0, 1]$, the following differential equation provides a MFA model of the evolution of the probabilities of infection of the nodes:
\vspace*{-.1in}
\begin{equation} \label{SISstatic}
\dot{p}_i(t) = (1 - p_i(t))\beta \sum_{j=1}^N W_{ij} p_j(t) - \delta p_i(t).
\end{equation}
This model provides a lower complexity deterministic approximation to the full dimension Markov process model of a SIS spread process evolving over a static network.  
Further details can be found in \cite{FallMMNP07},\cite{NowzariPP16},\cite{PBNTCNS18}. 
Discrete time versions of these approximation models have also been proposed and studied in \cite{AhnHCDC13},\cite{WangSRDS03}. 

The primary goals in most analyses of epidemic process dynamics include computing the system equilibria, and determining the convergence behavior of these processes near the equilibria.  Specifically, conditions for the existence of and convergence to ``disease-free'' or ``endemic'' equilibria are sought. For (\ref{SISstatic}), it is straightforward to see that the disease free state, $p_i^* = 0$ for all $i \in \{1, \ldots, N\}$, is a trivial equilibrium of the dynamics. It has been shown that this equilibrium 
is globally asymptotically stable if and only if $\frac{\beta}{\delta} \leq \frac{1}{\lambda_{max}(W)}$, where $\lambda_{max}(W)$ represents the largest real-valued part of the eigenvalue of the matrix $W$. It has further been shown, however, that if $\frac{\beta}{\delta} > \frac{1}{\lambda_{max}(W)}$, then there exists another equilibrium that is (almost) globally asymptotically stable, with $p_i^* \in (0,1)$ for all $i \in \{1, 2, \ldots, n \}$, implying the system converges asymptotically to an endemic state \cite{AhnHCDC13},\cite{KBGACC14},\cite{KBGCDC14},\cite{NPPCDC2014}. 

In this paper we consider a compartment model structure that specifically accounts for {\em infectious but asymptomatic} subgroups or individuals, namely a SAIR(S) model structure, incorporating Susceptible(S), Asymptomatic-infected(A), Infected-symptomatic(I), and Recovered(R) subsets of the population. We note that the asymptomatic subset we consider may include those individuals who do not experience symptoms through the course of their infection, as well as pre-symptomatic individuals. This structure may be used to directly capture the dynamics of COVID-19 and the role asymptomatic individuals play in the disease spread process; 
this model was first introduced for this purpose in public online seminars and panel discussions \cite{HCESCSummit},\cite{NeTs}, and in the literature in \cite{ARC-PBB}. Compartment models with different structures but including explicit asymptomatic population segments were previously proposed for dengue fever \cite{Grunnill2018} and rumor spreading over online social networks \cite{Zhu20}. We also note that simultaneous work has recently been completed is discussed in \cite{liu2020new}.

In the remainder of the paper, we first present the specific SAIR(S) group and networked models we will consider throughout, and discuss the equilibria and stability properties of these models in Section \ref{sec:SAIRS}.  We then discuss a series of simulation studies in Section \ref{sec:Simula} that illustrate our stability results as well as highlighting the role the asymptomatic subgroup plays in disease spread under various quarantine policies made with and without awareness of asymptomatic status. In Section \ref{sec:Estimate}, we provide a simple least squares estimation approach to compute the SAIR(S) model parameters from data, using local data (Champaign County Public Health District) for initial estimations; we discuss the challenges the currently available data present and our ongoing and future work in Section \ref{sec:Future}.


\vspace*{-.15in}
\section{The SAIRS model}\label{sec:SAIRS}

In order to investigate the effects of asymptomatic individuals on the spread of the epidemic, we consider the effects of a proportion of the infected subgroup being asymptomatic and potentially unaware of their carrier status. We evaluate both single group models as well as networked models, providing equilibria and stability analyses. 

\vspace*{-.1in}
\subsection{Single-Group and Networked Models}
Let $S(t),A(t),I(t),R(t)$,respectively, represent the proportion of susceptible, asymptomatic-infected, symptomatic-infected, and recovered individuals at time $t$. The Group SAIR(S) model we consider is characterized as:
\begin{align}
    &\dot{S}(t)=-\beta S(t)(A(t)+I(t))+\delta R(t)\notag\\
    &\dot{A}(t)=q\beta S(t)(A(t)+I(t))-\sigma A(t)-\kappa A(t)\label{eq:groupSAIR}\\
    &\dot{I}(t)=(1-q)\beta S(t)(A(t)+I(t)) +\sigma A(t)-\gamma I(t)\notag\\
    &\dot{R}(t)=\kappa A(t)+\gamma I(t)-\delta R(t)\notag
\end{align}
Here $\beta$ is the transmission rate amongst susceptible and infected groups, the latter of which includes both asymptomatic and symptomatic; $\kappa$ and $\gamma$, respectively, are the recovery rates for asymptomatic-infected and symptomatic-infected groups. Additionally, $q$ captures the proportion of individuals who are asymptomatic (and/or pre-symptomatic) but still infectious; correspondingly, $(1-q)$ represents the proportion of symptomatic individuals. Further, $\sigma$ is the progression rate from asymptomatic to symptomatic, and $\delta$ represents the rate at which immunity recedes. When $\delta=0$, individuals gain permanent immunity to the infection upon recovery. We assume these relations hold for all $t\geq 0$. \par

We also study the SAIR(S) model dynamics of n-subpopulations interconnected over an arbitrary network structure, with adjacency matrix denoted by $W$. Define $s_{i}, a_{i}, p_{i}, r_{i}$, respectively, as the proportion of the subpopulation $i$ that is susceptible (or healthy), asymptomatic-infected, symptomatic-infected, or recovered. The Networked SAIR(S) model (N-SAIR(S)) capturing the spread process over an arbitrary interconnection network is given by:
\begin{equation}\label{nsair-ct}
\begin{array}{lcl} 
\dot{s}_{i}(t)=-\beta_{i} s_{i}(t) \sum_{j}W_{ij}( a_{j}(t)+ p_{j}(t))+\delta_{i} r_{i}(t)\\
\dot{a}_{i}(t)=q\beta_{i} s_{i}(t) \sum_{j}W_{ij}( a_{j}(t)+ p_{j}(t))-\sigma_{i} a_{i}(t)-\kappa_{i}a_{i}(t)\\
\dot{p}_{i}(t)=(1-q)\beta_{i} s_{i}(t) \sum_{j}W_{ij}( a_{j}(t)+ p_{j}(t))+\sigma_{i} a_{i}(t)-\gamma_{i} p_{i}(t)\\
\dot{r}_{i}(t)=\kappa_{i} a_{i}(t)+\gamma_{i} p_{i}(t)-\delta_{i} r_{i}(t)
\end{array}
\end{equation}
where, similar to the Group Model (\ref{eq:groupSAIR}), for a subpopulation $i$, $\beta_{i}$ is the agent-to-agent transmission rate; $\kappa_{i}$ and $\gamma_{i}$, respectively, are the recovery rates for asymptomatic-infected and symptomatic-infected subsets; again, $\sigma_{i}$ represents the transition rate from asymptomatic to symptomatic infected; and $\delta_{i}$ represents the rate at which individuals may be susceptible to reinfection again after recovery. Since all individuals in a subgroup $i$ will reside in one of these subsets, we have  $s_{i}(t)+a_{i}(t)+p_{i}(t)+r_{i}(t)=1$, over all $i$, relative to the population size, $N_i$ of group $i$.\\
\textbf{Remark:} In the case where we have homogeneous spread parameters and the underlying network topology is complete with evenly distributed interconnection weights, that is, when $W_{ij}=1/n$ for all $i,j \in [n]$, and $(\beta_{i},\kappa_{i},\gamma_{i},\sigma_{i},\delta_{i})=(\beta, \kappa, \gamma, \sigma, \delta)$ for all $i \in [n]$, the Group Model (\ref{eq:groupSAIR}) and the Networked Model (\ref{nsair-ct}) coincide.

Prior to discussing the analysis of equilibria and stability for these models, we note the following result which establishes that the N-SAIR(S) model is well-defined. This result was first presented in \cite{ARC-PBB} for the discrete-time case using an induction argument; it is straightforward to adapt this result to the continuous-time model given in (\ref{nsair-ct}). We first state our assumption on the model parameters.

\begin{assumption}\label{assm:nsair-ct}
For all $i,j \in [n]$, we have $\beta_i,\;\gamma_i,\;\delta_i,\; \sigma_i,\; \delta_i,\;W_{ij}\geq0$, $0\leq q \leq 1$.

\end{assumption}

\begin{lemma} \label{lem:sair}
Consider the model in (\ref{nsair-ct}) under Assumption \ref{assm:nsair-ct}. Suppose $s_i(0), a_i(0), p_i(0),r_i(0) \in [0,1]$, $s_i(0) + a_i(0) +  p_i(0) + r_i(0) = 1$ for all $i \in [n] $. 
Then, for all $t \geq 0$ and $i \in [n]$, we have
$s_i(t), a_i(t), p_i(t), r_i(t) \in [0, 1]$ and $s_i(t) + a_i(t) + p_i(t) + r_i(t) = 1$. 
\end{lemma}

\begin{proof}
We show that for all $i \in [n]$, when $s_i(t)=0, \;a_i(t)=0, \;p_i(t)=0, \;r_i(t)=0$, then for $t\geq0$ we have $\dot{s_i}(t)\geq0, \;\dot{a_i}(t)\geq0, \;\dot{p_i}(t)\geq0, \;\dot{r_i}(t)\geq0$; and when $s_i(t)=1, \;a_i(t)=1, \;p_i(t)=1, \;r_i(t)=1$, then $\dot{s_i}(t)\leq0, \;\dot{a_i}(t)\leq0, \;\dot{p_i}(t)\leq0, \;\dot{r_i}(t)\leq0$.

From  $s_i(0)+a_i(0)+p_i(0)+r_i(0) = 1$, and $\dot{s_i}(t)+\dot{a_i}(t)+\dot{p_i}(t)+\dot{r_i}(t)= 0$, we have $s_i(t)+a_i(t)+ p_i(t)+r_i(t) = 1, \forall i \in [n], \forall t\geq0.$

By Assumption \ref{assm:nsair-ct} and (\ref{nsair-ct}), for all $i \in [n]$, if $s_i(0)=0$, we have $\dot{s_i}(0)=\delta_{i} r_{i}(0)\geq0$. By the continuity of $s_i(t)$, there exist $T_{s_i}\geq0$, such that, over the time interval $0\leq t \leq T_{s_i}$, $s_i(t)\geq0$. Similarly, $\dot{a_i}(0)=q\beta_{i} s_{i}(0) \sum_{j}W_{ij}( a_{j}(0)+ p_{j}(0))\geq 0$  if $a_i(0)=0$;  $\dot{p}_{i}(0)=(1-q)\beta_{i} s_{i}(0) \sum_{j}W_{ij}( a_{j}(0)+ p_{j}(0))+\sigma_{i} a_{i}(0)\geq 0$ if $p_i(0)=0$;  $\dot{r}_{i}(0)=\kappa_{i} a_{i}(0)+\gamma_{i} p_{i}(0)\geq 0$ if $r_i(0)=0$. Thus, there exist $T_{a_i}\geq0, T_{p_i}\geq0, T_{r_i}\geq0$, respectively, such that $a_i(t)\geq0$ for interval $0\leq t \leq T_{a_i}$; $p_i(t)\geq0$ for interval $0\leq t \leq T_{p_i}$; $r_i(t)\geq0$ over the time interval $0\leq t \leq T_{r_i}$.

Define $T_i:=min (T_{s_i},T_{a_i},T_{p_i},T_{r_i})$ for $i \in [n]$, and let $T=\min_{i \in [n]}T_i$. Then, at time $T, s_i(T)\geq0, a_i(T)\geq0, p_i(T)\geq0, r_i(T)\geq0, \forall i \in [n]$. Similar to the proof above, $\dot{s_i}(T)=\delta_{i} r_{i}(T)\geq0$  if $s_i(T)=0$; $\dot{a_i}(T)=q\beta_{i} s_{i}(T) \sum_{j}W_{ij}( a_{j}(T)+ p_{j}(T))\geq 0$  if $a_i(T)=0$; $\dot{p}_{i}(T)=(1-q)\beta_{i} s_{i}(T) \sum_{j}W_{ij}( a_{j}(T)+ p_{j}(T))+\sigma_{i} a_{i}(T)\geq 0$ if $p_i(T)=0$;  $\dot{r}_{i}(T)=\kappa_{i} a_{i}(T)+\gamma_{i} p_{i}(T)\geq 0$ if $r_i(T)=0$. Thus, for all $t\geq 0$ such that $s_i(t)=0, \;a_i(t)=0, \;p_i(t)=0$ or $r_i(t)=0$, $\dot{s_i}(t)\geq0, \;\dot{a_i}(t)\geq0, \;\dot{p_i}(t)\geq0, \;\dot{r_i}(t)\geq0$, respectively. This further suggests that, for all $i \in [n]$, $s_i(t)\geq0,\; a_i(t)\geq0,\; p_i(t)\geq0,\; r_i(t)\geq0$ for all $t\geq0$.

Next, we prove that $\dot{s_i}(t)\leq0, \;\dot{a_i}(t)\leq0, \;\dot{p_i}(t)\leq0, \;\dot{r_i}(t)\leq0$ when $s_i(t)=1, \;a_i(t)=1, \;p_i(t)=1, \;r_i(t)=1$, respectively. By Assumption \ref{assm:nsair-ct}, $s_i(t) + a_i(t) +  p_i(t) + r_i(t) = 1$, and $s_i(t), a_i(t), p_i(t), r_i(t)\geq0, \forall i \in [n]$, when $s_i(t)=1$, we have $a_i(t)=0, p_i(t)=0, r_i(t)=0$, which leads to $\dot{s_i}(t)=-\beta_{i}\sum_{j}W_{ij}( a_{j}(t)+ p_{j}(t))\leq0$. Similarly,  $\dot{a_i}(t)=-\sigma_{i}-\kappa_{i}\leq0$ when $a_i(t)=1$; $\dot{p_i}(t)=-\gamma_{i}\leq0$ when $p_i(t)=1$;  $\dot{r_i}(t)=-\delta_{i}\leq0$ when $r_i(t)=1$. Thus, we have, $s_i(t)\leq1,\; a_i(t)\leq1,\; p_i(t)\leq1,\; r_i(t)\leq1, \forall i \in [n], \forall t\geq 0.$
\end{proof}

\subsection{Equilibria and stability}
To quantitatively and qualitatively evaluate the propagation of the virus, the basic reproduction number, $R_{0}$, is a critical threshold quantity used widely in epidemiological studies. This number indicates how rapidly infected individuals transmit the virus to healthy individuals. In order to stop the virus from spreading exponentially, we want $R_{0}<1$. In this section, we evaluate the SAIR(S) system equilibria and conduct stability analysis around the equilibria, leading us to a stabilizing $R_{0}$ value. We first consider the group model. 

\vspace*{.1in}
\subsubsection{Group Model SAIRS}\vspace*{.1in}
\hfill\\
Noting that $S(t)=1-A(t)-I(t)-R(t)$, the nonlinear system (\ref{eq:groupSAIR}) can be written as:
\begin{small}
{\begin{align}
    &\dot{A}(t)=q\beta (1-A(t)-I(t)-R(t))(A(t)+I(t))-\sigma A(t)-\kappa   A(t)\notag\\
    &\dot{I}(t)=(1-q)\beta (1-A(t)-I(t)-R(t))(A(t)+I(t)) +\sigma A(t)-\gamma I(t)\label{eq:AIR}\\
    &\dot{R}(t)=\kappa A(t)+\gamma I(t) - \delta R(t)\notag
\end{align}}
\end{small}
By setting $\dot{A}(t), \dot{I}(t), \dot{R}(t)$ to $0$, we can see immediately that an equilibrium state of system (\ref{eq:AIR}) is given by $(A^e, I^e, R^e)=(0, 0, 0)$ with $S^e = 1$. This is the disease-free equlibrium (DFE) in the case of non-permanent immunity. 
Linearizing system (\ref{eq:AIR}) around $(A^e, I^e, R^e)$, we obtain the system Jacobian matrix, 
\begin{equation} \label{jacob1}
J^e = \left[ \begin{array}{ccc}
q\beta-\kappa-\sigma & q\beta & 0\\
(1-q)\beta+\sigma & (1-q)\beta-\gamma & 0\\
\kappa & \gamma & -\delta  \end{array} \right].
\end{equation}
Applying Theorem 4.7 from \cite{khalil}, we know the system will be globally asymptotically stable around the DFE if all eigenvalues of $J^e$ have negative real parts. Computing the characteristic polynomial for $J^e$, we have after some straightforward manipulations,
{\small \begin{equation} \label{stab_jacob}
\begin{array}{l} \det(\lambda I - J^e) = (\lambda + \delta) \cdot
\left[(\lambda -q\beta + \kappa + \sigma)(\lambda - (1-q)\beta + \gamma) - q(1-q)\beta^2 - q\beta \sigma \right] \end{array}
\end{equation}}
Applying the Routh-Hurwitz criterion to the second order polynomial in the second line of (\ref{stab_jacob}) gives us the following.
\begin{proposition}For the system given by (\ref{eq:AIR}), the DFE $(S^e, A^e, I^e, R^e) = (1, 0, 0, 0)$ is globally asymptotically stable (GAS) when
\begin{equation}\label{R01}
R_{0}:= \max\left( \frac{\beta}{\kappa+\gamma+\sigma}, \frac{\beta(q\gamma+(1-q)\kappa+\sigma)}{\gamma(\kappa+\sigma)}  \right) <1. 
\end{equation}
\end{proposition}

\vspace*{.1in}
Further, in the case where $\delta = 0$, that is when immunity following recovery from infection is permanent, the disease-free equilibria will be any points $(S^e, A^e, I^e, R^e)=(c_S, 0,0,c_R)$, where constants $c_R, c_S$ satisfy $c_S + c_R =1$. Analyzing the Jacobian for (\ref{eq:AIR}) in this case gives us that the equilibria $(S^e, A^e, I^e, R^e)=(c_S, 0,0,c_R)$ are globally asymptotically stable (GAS) again when (\ref{R01}) is satisfied.  That is, this basic reproduction number expression provides an appropriate threshold for determining when the spread process for the SAIR(S) model will or will not spread exponentially in either of the scenarios of permanent or non-permanent immunity.

We further consider the case where the asymptomatic-infected and symptomatic-infected individuals have different infection transmission rates.
In the case of COVID-19, this difference could be partly due to the inability to conduct large-scale population testing allowing us to efficiently identify and isolate Asymptomatic individuals. Thus, we would have different quarantine control effectiveness over these two subpopulations.

In this case, we denote the infection transmission rates for agent-to-agent contact between the susceptible subgroup and the two infectious groups, resp., as $\beta_{A}, \beta_{I}$. As in the preceding analysis, we compute the Jacobian around the disease-free equilibrium $(S^e, A^e, I^e, R^e) = (1,0,0,0)$, as
\vspace*{-.1in}
\begin{equation} \label{jacob2}
J^e = \left[ \begin{array}{ccc}
q\beta_{A}-\kappa-\sigma & q\beta_{I} & 0\\
(1-q)\beta_{A}+\sigma & (1-q)\beta_{I}-\gamma & 0\\
\kappa & \gamma & -\delta
\end{array} \right].
\end{equation}

\vspace*{-.1in}
Following a similar approach as before yields:
\vspace*{-.1in}
\begin{small}
\begin{equation}\label{R03}
    R_{0} := \max \left( \frac{q\beta_{A}+(1-q)\beta_{I}}{\kappa+\gamma+\sigma},\frac{q\beta_{A}\gamma+\beta_{I}((1-q)\kappa+\sigma)}{\gamma(\kappa+\sigma)} \right)
\end{equation}
\end{small}

\vspace*{-.1in}
For GAS, again it is required that $R_0 < 1$.  

We can further compute an endemic equilibrium point for (\ref{eq:AIR}). In this case, we assume non-permanent immunity, that is, $\delta > 0$. Setting $\dot{A}(t), \dot{I}(t), \dot{R}(t)$ to $0$, we obtain the endemic equilibrium,  

\begin{small}
\begin{equation} \label{eq:endemic1}
\left[ \begin{array}{c} S^e \\ A^e \\ I^e \\ R^e \end{array} \right] = \left[ \begin{array}{c} \frac{\gamma(\kappa+\sigma)}{\beta(q\gamma+(1-q)\kappa+\sigma)}\\
\frac{q\delta\gamma\Big(\beta(q\gamma+(1-q)\kappa+\sigma)-\gamma(\kappa+\sigma)\Big)}{\beta(q\gamma+(1-q)\kappa+\sigma)\Big(\gamma(\kappa+\sigma)+\delta(q\gamma+(1-q)\kappa+\sigma)\Big)} \\
\frac{\delta((1-q)\kappa+\sigma)\Big(\beta(q\gamma+(1-q)\kappa+\sigma)-\gamma(\kappa+\sigma)\Big)}{\beta(q\gamma+(1-q)\kappa+\sigma)\Big(\gamma(\kappa+\sigma)+\delta(q\gamma+(1-q)\kappa+\sigma)\Big)}\\
\frac{\gamma(\kappa+\sigma)\Big(\beta(q\gamma+(1-q)\kappa+\sigma)-\gamma(\kappa+\sigma)\Big)}{\beta(q\gamma+(1-q)\kappa+\sigma)\Big(\gamma(\kappa+\sigma)+\delta(q\gamma+(1-q)\kappa+\sigma)\Big)} \end{array} \right]
\end{equation}
\end{small}

Denoting $\Psi =\gamma(\kappa+\sigma)$ and  $\Phi=q\gamma+(1-q)\kappa+\sigma$ and noting both $\Psi >0$ and $\Phi >0$, and further defining \begin{align*}
     C&=\beta\Phi-\Psi=\beta(q\gamma+(1-q)\kappa+\sigma)-\gamma(\kappa+\sigma),\\
     D&=\delta\Phi+\Psi=\delta(q\gamma+(1-q)\kappa+\sigma)+\gamma(\kappa+\sigma)>0,
\end{align*}
then the endemic equilibrium can be written as
\begin{equation*} 
\left[\begin{array}{c} S^e \\ A^e \\ I^e \\ R^e \end{array} \right] = \left[\arraycolsep=1.4pt\def\arraystretch{1.6} \begin{array}{c} \frac{\Psi}{\beta\Phi}\\
\frac{q\delta\gamma(\beta\Phi-\Psi)}{\beta\Phi(\delta\Phi+\Psi)} \\
\frac{\delta((1-q)\kappa+\sigma)(\beta\Phi-\Psi)}{\beta\Phi(\delta\Phi+\Psi)}\\
\frac{\Psi(\beta\Phi-\Psi)}{\beta\Phi(\delta\Phi+\Psi)} \end{array} \right]= \left[\arraycolsep=1.4pt\def\arraystretch{1.6} \begin{array}{c} \frac{\Psi}{\beta\Phi}\\
\frac{q\delta\gamma C}{\beta\Phi D} \\
\frac{\delta((1-q)\kappa+\sigma)C}{\beta\Phi D}\\
\frac{\Psi C}{\beta\Phi D} \end{array} \right].
\end{equation*}
From this expression for the endemic equilibrium point we compute the Jacobian around the equilibrium point as:
\begin{equation*} 
J^e = \left[\arraycolsep=3pt\def\arraystretch{1.6} \begin{array}{ccc}
-\frac{(\kappa+\sigma)\big((1-q)\kappa+\sigma)\big)}{\Phi}-\frac{q\delta C}{D}  & \frac{q\Psi}{\Phi}-\frac{q\delta C}{D} & -\frac{q\delta C}{D}\\
\frac{(\gamma+\sigma)\big((1-q)\kappa+\sigma)\big)}{\Phi}-\frac{(1-q)\delta C}{D}  & -\frac{q\gamma(\gamma+\sigma)}{\Phi}-\frac{(1-q)\delta C}{D} & -\frac{(1-q)\delta C}{D}\\
\kappa & \gamma & -\delta
\end{array} \right].
\end{equation*}
Letting $F=\gamma+\kappa+\sigma$, for GAS of the endemic equilibrium we require: $$C>0,$$ and 
\begin{footnotesize}
\begin{equation*}
    \frac{(CD+D^2)(\delta+F)\Phi(F\Phi-\Psi)-D^2\Psi(F\Phi-\Psi)+\delta(\delta+F)(C^2+CD)\Phi^2+\delta CD\Phi(F\Phi-\Phi^2-\Psi)}{CD\Psi\Phi^2}>1.
\end{equation*}
\end{footnotesize}

\subsubsection{Networked Model N-SAIR(S)}
\hfill\\
We now evaluate equilibria and their stability properties for the networked SAIR(S) models. First we consider the case with permanent immunity, i.e., $\delta=0$. Given $s_{i}(t)=1-a_{i}(t)-p_{i}(t)-r_{i}(t)$ for all $t\geq0,\; i\in[n]$, system (\ref{nsair-ct}) can be represented in matrix form as:
\begin{small}
\begin{align}
    \dot{a(t)}&=[q(I-A(t)-P(t)-R(t))BW-\Sigma-K]a(t)+q(I-A(t)-P(t)-R(t))BW p(t)\notag\\
    \dot{p(t)}&=[(1-q)(I-A(t)-P(t)-R(t))BW +\Sigma]a(t)+[(1-q)(I-A(t)-P(t)-R(t))BW-\Gamma] p(t)\label{NSAIRSmatrix}\\
    \dot{r(t)}&= Ka(t)+\Gamma p(t).\notag
\end{align}
\end{small}
Here, 
\begin{small}
\begin{equation*} 
a(t)=\left[ \begin{array}{c} 
a_{1}(t)\\
\vdots\\
a_{n}(t)
\end{array} \right], \;
p(t)=\left[ \begin{array}{c} 
p_{1}(t)\\
\vdots\\
p_{n}(t)
\end{array} \right], \;
r(t)=\left[ \begin{array}{c} 
r_{1}(t)\\
\vdots\\
r_{n}(t)
\end{array} \right],
\end{equation*}
\end{small}
with $n\times n$ matrices $A(t)=diag(a_{i}(t)), \; P(t)=diag(p_{i}(t)), \; R(t)=diag(r_{i}(t)), \; B=diag(\beta_{i}), \; K=diag(\kappa_{i}), \; \Gamma=diag(\gamma_{i}), \; \Sigma=diag(\sigma_{i}), \; \Delta=diag(\delta_{i})$, and adjacency matrix $W$.\par

Setting $\dot{a}(t),\dot{p}(t),\dot{r}(t)$ to 0, we can compute the equilibrium state where $a^e=p^e=\overline{0},r^e=\overline{r_{c}}$, where $\overline{r_{c}}$ is any non-negative constant vector with elements $r_{c_i}< 1$. Linearizing the system (\ref{NSAIRSmatrix}) at the equilibrium $(a^e, p^e,r^e)$, we obtain the $3n\times 3n$ system Jacobian Matrix given by
\begin{equation} \label{jacob3}
J^e = \left[ \begin{array}{ccc}
q(I-R_{c})BW-\Sigma-K & q(I-R_{c})BW & 0\\
(1-q)(I-R_{c})BW+\sigma & (1-q)(I-R_{c})BW-\Gamma & 0\\
K & \Gamma & -\Delta
\end{array} \right].
\end{equation}
Analysis of this Jacobian matrix leads to a set of constraints on the spectrum of the weighting matrix $W$. An alternative approach is to consider a Lyapunov stability analysis approach. 
Here we consider a quadratic Lyapunov function $V=a^{T}B^{-1}a+p^{T}B^{-1}p$. Computing the derivative, we can easily show 
\begin{equation}
\dot{V} \leq a^{T}[qW-B^{-1}(\Sigma+K)]a+p^{T}[(1-q)W-b^{-1}\Gamma]p+a^{T}(W+B^{-1}\Sigma)p.
\end{equation}
For GAS, we require $\dot{V}<0$ for all $t\geq0$, which after some algebraic manipulations can be shown to be equivalent to the inequality
\begin{small}
\begin{equation}\label{Lyapunov}
    \left[ \begin{array}{cc}a^{T}&p^T\end{array}\right]\left[\begin{array}{cc}qW&\frac{1}{2}W\\\frac{1}{2}W & (1-q)W\end{array}\right]\left[\begin{array}{c}a\\p\end{array}\right]<\left[\begin{array}{cc}a^{T}&p^T\end{array}\right]\left[\begin{array}{cc}B^{-1}(\Sigma+K)&-\frac{1}{2}B^{-1}\Sigma\\-\frac{1}{2}B^{-1}\Sigma& B^{-1}\Gamma\end{array}\right]\left[\begin{array}{c}a\\
    p\end{array}\right].
\end{equation}
\end{small}
Applying the Rayleigh quotient we then have the following sufficient condition for the DFE:
{\small
\begin{equation} \label{NSAIRS-stability1}
\left[ \begin{array}{cc}
qW & \frac{1}{2}W\\
\frac{1}{2}W & (1-q)W\end{array} \right] \prec \left[ \begin{array}{cc}
B^{-1}(\Sigma+K) & -\frac{1}{2}B^{-1}\Sigma\\
-\frac{1}{2}B^{-1}\Sigma & B^{-1}\Gamma\end{array}\right],
\end{equation}}
where $\prec$ denotes relative definiteness of the matrices. That is, (\ref{NSAIRS-stability1}) provides a test that bounds the maximum eigenvalue of the $q$-scaled adjacency matrix $W$ in terms of the minimum eigenvalue of a matrix consisting of  diagonal block entries of ratios of healing and transition rates ($\kappa_i$, $\gamma_i$ and $\sigma_i$) to infection rates ($\beta_i$); this loosely generalizes 
the usual $R_0$ threshold to allow for heterogeneous infection parameters over multiple infection compartments.

\section{Simulations} \label{sec:Simula}
In this section, we first illustrate the role of the asymptomatic subgroup in the development of the epidemic, and then investigate the effects of quarantine and social distancing policies.\par
First, we simulate a baseline group/networked model (\ref{nsair-ct}), for which we assume homogeneous spread parameters and a five-subpopulation network structure 
We assume the total population size is $10,000$ and the respective subpopulations denoted A, B, C, D, and E have populations $2000, 2500, 1500, 3500$, and $500$, respectively. We assume the cities are fully connected with evenly distributed edge weights, thus this baseline model is equivalent to a single group model. 
We use the estimation results from local data (discussed further in Section \ref{sec:Estimate}) in addition to drawing upon the literature (e.g.,\cite{StephenA} \cite{Oran}) on COVID-19 to inform our parameter value selection, where we set $(q,\beta, \sigma, \gamma, \kappa, \delta)=(0.7, 0.25, 0.15, 0.11, 0.08, 0.0001).$
We further set the initial proportions of the $A, I, R$ compartments as
\begin{align*}
    a(0)&=(a_{A}(0),a_{B}(0),a_{C}(0),a_{D}(0),a_{E}(0))=(0.006,0.004,0.012,0.004,0.004)\notag\\
    p(0)&=(p_{A}(0),p_{B}(0),p_{C}(0),p_{D}(0),p_{E}(0))=(0.005,0.002,0.008,0.003,0.002)\notag\\
    r(0)&=(r_{A}(0),r_{B}(0),r_{C}(0),r_{D}(0),r_{E}(0))=(0.007,0.003,0.010,0.008,0.005)\notag
\end{align*}

Simulating the SAIRS model over $60$ days results in the epidemic progression shown in Fig.\ref{fig:Fig1}:
\begin{figure}[h]
\centering
\includegraphics[scale=0.35]{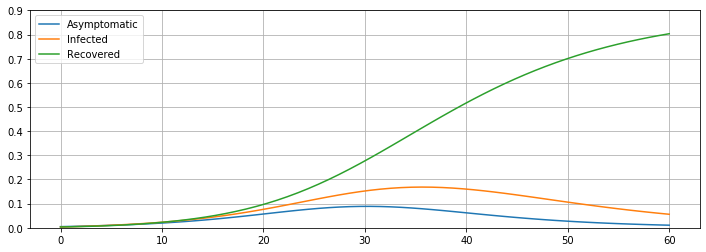}
\caption{Group/Network SAIRS Simulation: Baseline Model}
\label{fig:Fig1}
\end{figure}\\
Note that peak active infection occurs on day $33$, that is 
$p(t)+a(t)$ attains a maximum of approximately $28\%$ on day $t=33$. By day $60$, approximately $87\%$ of the entire population has been or is infected; assuming a mortality rate of $4\%$, that would correspond to $348$ deaths in the two month time span. Again we note this model assumes homogeneous mixing within the entire population. 

\subsection{Endemic equilibrium}
From Section \ref{sec:SAIRS}, we know that the healthy-state equilibrium is GAS given that 
\begin{equation}\label{R03}
    R_{0} := \max \left(\frac{q\beta _{A}+(1-q)\beta_{I}}{\kappa+\gamma+\sigma},\frac{q\beta_{A}\gamma+\beta_{I}((1-q)\kappa+\sigma)}{\gamma(\kappa+\sigma)} \right)<1.
\end{equation}
Alternatively when $\frac{q\beta_{A}\gamma+\beta_{I}((1-q)\kappa+\sigma)}{\gamma(\kappa+\sigma)}>1$, the system converges to the endemic equilibrium.

Setting the initial conditions as in the healthy-state equilibrium scenario, excepting a change in the parameter value of $delta$ from $0.0001$ to $0.01$, we simulate the SAIRS model over 500 days. The resulting infection progression plot is given in Figure \ref{fig:endemic}.
\begin{figure}[h]
\centering
\includegraphics[scale=0.35]{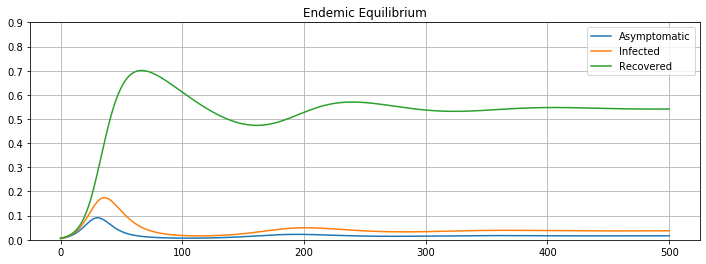}
\caption{Group/Network SAIRS Simulation: Endemic Equilibrium}
\label{fig:endemic}
\end{figure}
We note that peak infection happens on day $34$ with approximately $26\%$ of population being infected (both asymptomatic and symptomatic). At the endemic equilibrium, approximately $5.4\%$ of the population is actively infected. 

\subsection{Asymptomatic Effects}
One major obstacle in the control of COVID-19 is the challenge of identifying and monitoring individuals in the asymptomatic-infected subgroup. Herein we explore the impact of the asymptomatic subgroup on the epidemic evolution.

First, we assume there is no control imposed on either the asymptomatic or symptomatic infected subgroups. For simplicity, we use the group model (\ref{eq:groupSAIR}) with parameters
\begin{small}
\begin{equation}
    (q, \beta, \sigma, \gamma, \kappa, \delta)=(0.7, 0.25, 0.15, 0.11, 0.08, 0.0001)\notag
\end{equation}
\end{small}
which gives a basic reproduction number $R_{0} \approx 2.5$ from (\ref{R01}).\\
By setting initial proportions for the $A,I,R$ compartments for each subpopulation as 
\begin{equation*}
    (a(0),I(0),R(0))=(0.004,0.002,0.003),
\end{equation*}
we obtain the simulation results shown in Fig. \ref{fig:Fig 4}.
\begin{figure}[h!]
\centering
\includegraphics[scale=0.32]{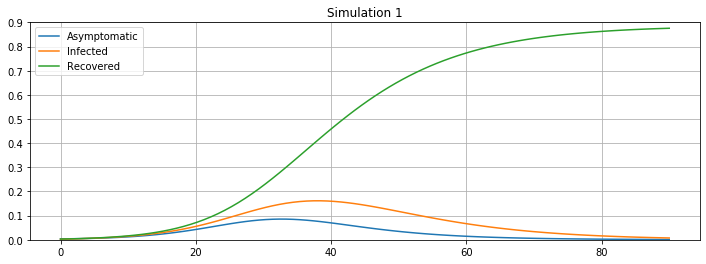}
\caption{No control policies in effect on either Asymptomatic or Symptomatic Infected subgroups}
\label{fig:Fig 4}
\end{figure}
The population reaches a peak infection level of approximately $25\%$ on day $35$. By day $80$, approximately $87\%$ of the population has been or is infected.

Next, we implement moderate and stringent isolation policies on only the symptomatic subgroup; this is effected in the simulations by changing the respective infection rate parameters of the subgroups, which we now denote individually by $\beta_A$ and $\beta_I$.  Imposing isolation policies on a subgroup effectively lowers the corresponding infection rate. The simulations results are shown in Fig. \ref{fig:Fig5}
\begin{figure}[h]
\begin{subfigure}{.5\textwidth}
  \centering
  \includegraphics[scale=0.32]{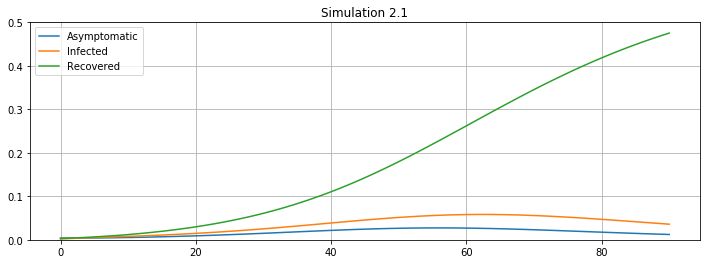}  
  \caption{Moderate isolation of the Symptomatic Infected subgroup; $\beta_{A}=0.25, \beta_{I}=0.11$ giving $R_{eff}=1.5$}
  \label{fig:5(a)}
\end{subfigure}
\begin{subfigure}{.5\textwidth}
  \centering
  \includegraphics[scale=0.32]{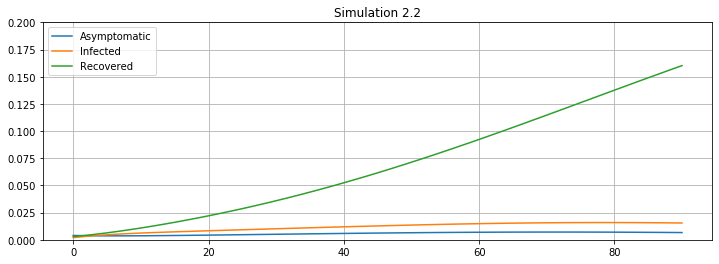}
  \caption{Stringent isolation of the Symptomatic Infected subgroup; $\beta_{A}=0.25, \beta_{I}=0.06$ giving $R_{eff}=1.2 $}
  \label{fig:5(b)}
\end{subfigure}
\caption{Imposing isolation policies on subgroup $I$}
\label{fig:Fig5}
\end{figure}
We note that with isolation measures on only the symptomatic infected subgroup, the epidemic now progresses more slowly and mildly, as is expected, however there is still substantial infection in the population. The infection peaks at days $60$ and $75$, respectively, approximately $4-6$ weeks later than with no control. With moderate isolation policies in effect on the $I$ subgroup, the peak infection level is approximately $9\%$ and with strict isolation policies the peak infection level attained is approximately $2.5\%$. Finally by day $80$, the total percentages of the population that has been or is infected is approximately $49\%$ and $17\%$; with a mortality rate of $4\%$ this corresponds to $196$ and $68$ deaths, respectively.\par

Alternatively, we consider the situation where Asymptomatic individuals are also identified and isolated, under both moderate and stringent policies, with the results shown in Fig \ref{fig:Fig6}.
\begin{figure}[h]
\begin{subfigure}{.5\textwidth}
  \centering
  \includegraphics[scale=0.32]{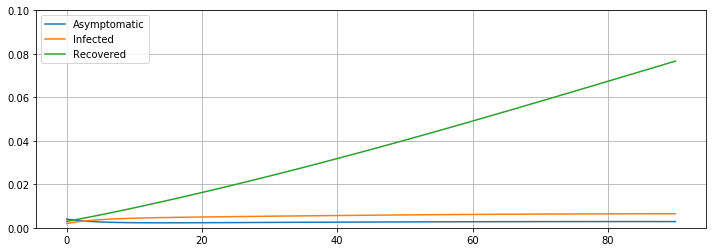}  
  \caption{Moderate isolation of both Symptomatic and Asymptomatic Infected subgroups; $\beta_{A}=0.11, \beta_{I}=0.11$ giving $R_{eff}=1.09$}
  \label{fig:6(a)}
\end{subfigure}
\begin{subfigure}{.5\textwidth}
  \centering
  \includegraphics[scale=0.32]{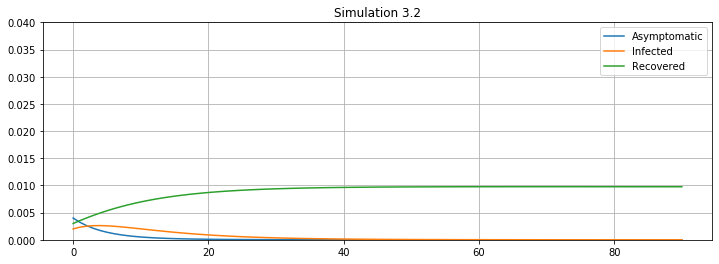}
  \caption{Stringent isolation of both Symptomatic and Asymptomatic Infected subgroups; $\beta_{A}=0.0125, \beta_{I}=0.0125$ giving $R_{eff}=0.12$}
  \label{fig:6(b)}
\end{subfigure}
\caption{Imposing isolation policies on subgroups $A$ and $I$}
\label{fig:Fig6}
\end{figure}\par
Note that, with only moderate isolation on both Asymptomatic and Symptomatic Infected groups (\ref{fig:6(a)}), the epidemic is under control within three months. By day $80$, approximately $7.7\%$ of the population has been or is infected, corresponding to a total of $770$ individuals in a population base of $10,000$ that have been infected; at a $4\%$ mortality rate this corresponds to approximately $31$ deaths as compared to approximately $68$ deaths with stringent control imposed on only the Symptomatic Infected group (\ref{fig:5(b)}). 

An additional perspective to consider is the effective reproduction number under the different isolation policies. Moderate isolation of both Asymptomatic and Symptomatic subgroups (\ref{fig:6(a)}) gives a $R_{eff}\approx 1.09$, while stringent isolation on just the Symptomatic subgroup (\ref{fig:5(b)}) gives a $R_{eff}\approx 1.2$.\par

These simulation results indicate that identification and isolation of Asymptomatic infected individuals is much more effective in curbing the spread of the epidemic than identification and isolation of just the Symptomatic subgroup. However, due to the voluntary nature of most testing regimens in the United States, this type of basic control has not been implemented. To achieve this goal, either regular extensive mandatory testing policies, or persistent extensive isolation of the whole population is required. 

\subsection{Network Effects}
Here, we evaluate the effects that a more realistic interaction structure has on epidemic spread over a population. We consider the $5$-node network introduced earlier, and consider the removal of some edges between nodes.  We first consider an interconnection network structure with adjacency matrix
{\small
\begin{equation} \label{Adjaciency Matrix}
W = \left[ \begin{array}{ccccc}
\frac{1}{3} & \frac{1}{3} & \frac{1}{3} & 0 & 0\\
\frac{1}{3} & \frac{1}{3} & 0 & 0 & \frac{1}{3}\\
\frac{1}{3} & 0 & \frac{1}{3} & \frac{1}{3} & 0\\
0 & 0 & \frac{1}{2} & \frac{1}{2} & 0\\
0 & \frac{1}{2} & 0 & 0 & \frac{1}{2}
\end{array} \right].
\end{equation}}
Using the same parameters and initial conditions as in the baseline model, our simulations return results for subpopulations C and E as shown in Fig.\ref{fig:Fig2}.
\begin{figure}[h]
\begin{subfigure}{.5\textwidth}
  \centering
  \includegraphics[scale=0.32]{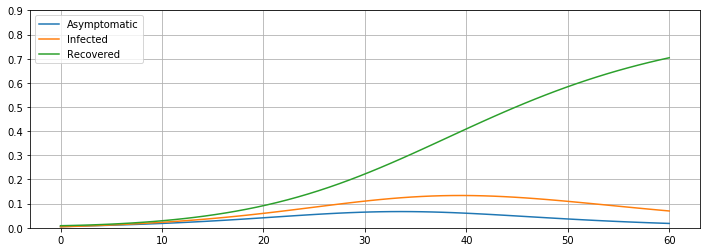}  
  \caption{Subpopulation C}
  \label{fig:sub-first}
\end{subfigure}
\begin{subfigure}{.5\textwidth}
  \centering
  \includegraphics[scale=0.32]{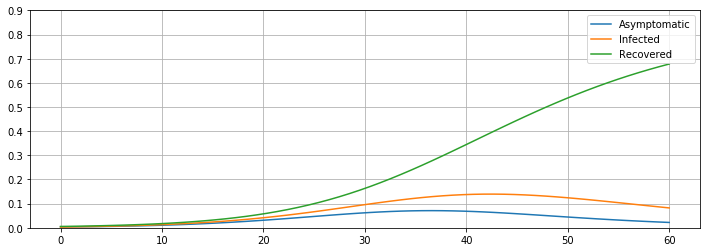}
  \caption{Subpopulation E}
  \label{fig:sub-second}
\end{subfigure}
\caption{Strongly Connected Network Simulation Results}
\label{fig:Fig2}
\end{figure}\par
\begin{figure*}[hbt!]
  \centering
  \subcaptionbox{Subpopulation 4}[.3\linewidth][c]{%
    \includegraphics[width=\linewidth]{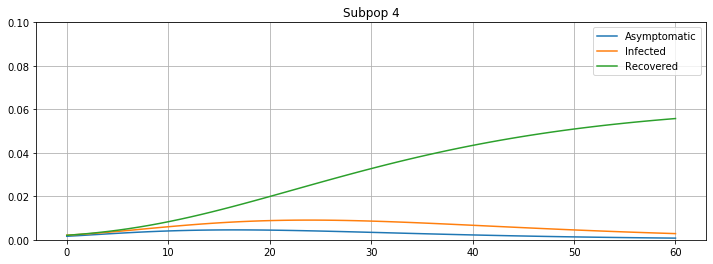}}\quad
  \subcaptionbox{Subpopulation 22}[.3\linewidth][c]{%
    \includegraphics[width=\linewidth]{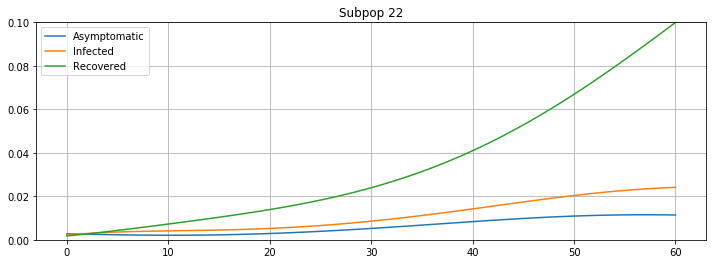}}\quad
  \subcaptionbox{Subpopulation 28}[.3\linewidth][c]{%
    \includegraphics[width=\linewidth]{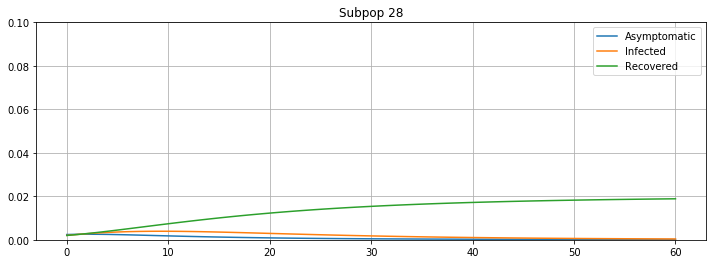}}

  \medskip

  \subcaptionbox{Subpopulation 32}[.3\linewidth][c]{%
    \includegraphics[width=\linewidth]{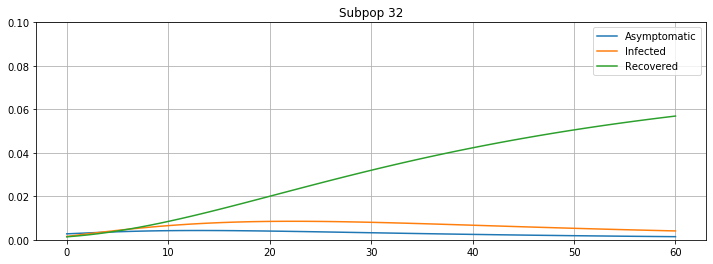}}\quad
  \subcaptionbox{Subpopulation 34}[.3\linewidth][c]{%
    \includegraphics[width=\linewidth]{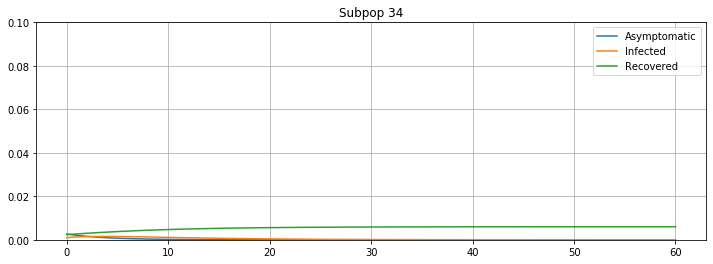}}\quad
  \subcaptionbox{Subpopulation 48}[.3\linewidth][c]{%
    \includegraphics[width=\linewidth]{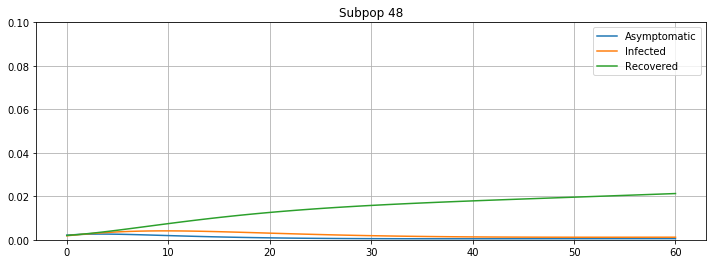}}
  \caption{Weakly Connected Network Simulation Results}
  \label{fig:Fig3}
\end{figure*}
With a less strongly connected network, the epidemic spreads more slowly and weakly. Subpopulation C reaches its peak infection level at day $37$, and subpopulation E at day $39$. By day $60$, approximately $83\%$ of city C population and $81\%$ city E population have been infected. However, in total, approximately $480$ fewer individuals over the five cities are infected as compared to the fully connected (i.e., complete) baseline model.\par

Next, to explore the impact of quarantine and stronger social distancing measures, we break the full population into $50$ smaller subpopulations, and generate a stochastic adjacency matrix with each node only connected to (randomly selected) $20$ nodes out of the total of $50$ group nodes.
We also generated the initial conditions randomly, i.e., $a(0), p(0), r(0)$, assuming $a_{i}(0)\sim \mathcal{N}(0.04,\,0.005)$, $p_{i}(0)\sim \mathcal{N}(0.02,\,0.005)$, $r_{i}(0)\sim \mathcal{N}(0.03,\,0.005)$, with these values restricted to be non-negative.
Randomly selecting $6$ of the $50$ sub-populations, we can demonstrate the simulation results as shown in Fig.\ref{fig:Fig3}:\par

Note that, with this more extensive isolation, the epidemic decays much faster than under the previous strongly connected network (Fig.\ref{fig:Fig2}). Subpopulations $4, 22, 28, 32, 34$ and $48$, respectively, reach their peak infection levels at days $21, 59, 7, 24, 0$ and day $8$. Among the six subpopulations in the sample, subpopulation $22$ is the most highly infected group. However, overall after $60$ days, approximately only $13.6\%$ of the population has been or is infected, which is a reduction of $73.4\%$ of the population compared to the fully connected network (Fig.\ref{fig:Fig1}), and a reduction of $67.7\%$ compared to the strongly connected network (Fig.\ref{fig:Fig2}). These simulations demonstrate that social distancing measures, such as quarantining within each community or family, does serve to control the spread of the epidemic, as has been seen in practice in many communities. From the perspective of the group model, extensive isolation policies help reduce the group transmission rate for person-to-person contact, which results in both faster flattening of the infection curve and fewer infected individuals (asymptomatic and symptomatic) in the whole population.

\section{Parameter estimation} \label{sec:Estimate}
As was noted in previous sections, we determined values for the N-SAIR(S) model parameters based on both estimation results using local data and similarly estimated values from the literature. In this section we discuss the simple least-squares approach we use for parameter estimation for a discrete-time N-SAIRS model, given in (\ref{NSAIRS-discrete}). We further present some of our initial estimation results using local data for COVID-19.

As our data comes from sampling on a daily basis, a discrete-time model would be better suited for estimating and evaluating model parameters. Thus, we first apply a forward Euler's method to the continuous-time networked system (\ref{nsair-ct}), giving us a discrete-time networked SAIRS model,
\begin{small}
\begin{align}
    &a_{i}^{k+1}=a_{i}^{k}+q\beta_{i}(1-a_{j}^k- p_{j}^k-r_{j}^k) \sum_{j}W_{ij}( a_{j}^k+ p_{j}^k)-\sigma_{i} a_{i}^{k}-\kappa_{i}a_{i}^{k}\notag\vspace{-.1in}\\
    &p_{i}^{k+1}=p_{i}^{k}+(1-q)\beta_{i}(1-a_{j}^k- p_{j}^k-r_{j}^k) \sum_{j}W_{ij}( a_{j}^k+ p_{j}^k)+\sigma_{i} a_{i}^{k}-\gamma_{i} p_{i}^{k}\notag{-.1in}\\
    &r_{i}^{k+1}=r_{i}^{k}+\kappa_{i} a_{i}^{k}+\gamma_{i} p_{i}^{k} -\delta_{i} p_{i}^{k}. \label{NSAIRS-discrete}
\end{align}
\end{small}
Since our simulation update will be daily and the sampling rate is once-per-day, the sampling parameter typically made explicit in sampled-data models will be $1$ and thus is not explicitly noted above.

\subsection{Asymptomatic Proportion Estimation}
Due to the difficulties in identifying and monitoring infected individuals without symptoms, explicit and unbiased information for asymptomatic-infected estimations is not always available. Applying the Next-Day Law approach proposed by Nesterov in \cite{Nesterov}, we estimate asymptomatic numbers per day, based on a latent period assumption, and further estimate the proportion $q$ of the asymptomatic subpopulation as a fraction of the total population. We note that this approach more accurately gives us the pre-symptomatic subpopulation proportion. We include the Next-Day Law here for completeness.

\begin{proposition}
\cite{Nesterov} Let $T(d)$ represent the total number of confirmed cases by day $d$, and $A(d)$ represent the number of asymptomatic infected individuals at the beginning of day $d$. Assume the latent period (that is the time from exposure to onset of symptoms) is a constant time of $\Delta$ days. Then,  $A(d+1)=T(d+\delta)-T(d), \forall d \in \mathbb{Z}$
\end{proposition}\par
From the estimated daily asymptomatic population, we are able to estimate the proportions $q$ and $1-q$ of asymptomatic and symptomatic-infected subgroups. Using the first two equations in (\ref{NSAIRS-discrete}), and omitting the linear terms, we have 
\[
\frac{a_{i}^{k+1}-a_{i}^{k}}{p_{i}^{k+1}-p_{i}^{k}}\approx\frac{q}{1-q}.
\]
Hence, we can simply estimate $q$ as $\frac{a_{i}^{k+1}-a_{i}^{k}}{(a_{i}^{k+1}-a_{i}^{k})+(p_{i}^{k+1}-p_{i}^{k})}$. 

\subsection{Least squares estimation of model parameters}
With $q$ known or estimated, we can apply the simple approach first outlined in \cite{pare2018dtjournal}, and further described for SAIRS models in \cite{ARC-PBB} to estimate the model parameters $\beta_{i}, \sigma_{i}, \kappa_{i}, \gamma_{i}$, and $\delta_{i}$. We first rewrite the networked system (\ref{NSAIRS-discrete}) in matrix form. Let

\begin{align}\label{pseudoinverse}
{b} &= \left[ \begin{array}{ccccccccc}
a_{i}^{1}-a_{i}^{0}\\
\vdots\\
a_{i}^{T}-a_{i}^{T-1}\\
p_{i}^{1}-p_{i}^{0}\\
\vdots\\
p_{i}^{T}-p_{i}^{T-1}\\
r_{i}^{1}-r_{i}^{0}\\
\vdots\\
r_{i}^{T}-r_{i}^{T-1}\\
\end{array} \right], \; 
{\cal A}=\left[ \begin{array}{ccc}
\Phi_{i}\\
\Sigma_{i}\\
\Gamma_{i}
\end{array} \right], \; 
x=\left[ \begin{array}{ccccc}
\beta_{i}\\
\sigma_{i}\\
\gamma_{i}\\
\kappa_{i}\\
\delta_{i}
\end{array} \right],\notag\\
\end{align}
with
\begin{small}
\begin{align} \label{pseudo2}
\Phi_{i}&=\left[ \begin{array}{ccccc}
qs_{i}^0\sum_{j}W_{ij}( a_{j}^0+ p_{j}^0) & -a_{i}^0 & 0 & -a_{i}^0 & 0\\
\vdots & \vdots & \vdots & \vdots & \vdots\\
qs_{i}^{T-1}\sum_{j}W_{ij}( a_{j}^{T-1}+ p_{j}^{T-1}) & -a_{i}^{T-1} & 0 & -a_{i}^{T-1} & 0
\end{array} \right],\notag\\
\Sigma_i&=\left[ \begin{array}{ccccc}
(1-q)s_{i}^0\sum_{j}W_{ij}( a_{j}^0+ p_{j}^0) & a_{i}^0 & -p_{i}^0 & 0 & 0\\
\vdots & \vdots & \vdots & \vdots & \vdots\\
(1-q)s_{i}^{T-1}\sum_{j}W_{ij}( a_{j}^{T-1}+ p_{j}^{T-1}) & a_{i}^{T-1} & -p_{i}^{T-1} & 0 & 0
\end{array} \right],\notag\\
\Gamma_i&=\left[ \begin{array}{ccccc}
0 & 0 & a_{i}^0 & p_{i}^0 & -r_{i}^0\\
\vdots & \vdots & \vdots & \vdots & \vdots\\
0 & 0 & a_{i}^{T-1} & p_{i}^{T-1} & -r_{i}^{T-1},
\end{array} \right]\notag
\end{align}
\end{small}
where $s_{i}^k=1-a_{i}^k-p_{i}^k-r_{i}^k, \forall i \in [n], k \in \mathbb{Z}$.\\
Then the discrete-time networked SAIRS model can be written in the form of a system of linear equations,
\begin{equation}\label{eq}
{b} = {\cal A} x  \quad\quad \forall i \in [n]
\end{equation}\par
That is, since $q$ is assumed known, (\ref{eq}) is linear with respect to the remaining model parameters. When ${\cal A}$ is full rank, we can thus recover the parameters $\beta_{i}^*, \sigma_{i}^*, \gamma_{i}^*, \kappa_{i}^*$,  and $\delta_{i}^*$ using the pseudo-inverse in the least-squares solution to (\ref{eq}). 

\subsection{Preliminary estimation results}
We consider local COVID-19 testing-site data from Champaign County, Illinois, dating from May to August, to conduct parameter estimations for different phases of state restore plans, scheduled as:
\begin{align}
    &Phase \; 1: Rapid \; Spread \; (04/01/2020-05/01/2020)\notag\\
    &Phase \; 2: Flattening \; (05/01/2020-05/29/2020)\notag\\
    &Phase \; 3: Recovery \; (05/29/2020-06/26/2020)\notag\\
    &Phase \; 4: Revitalization \; (06/26/2020-09/26/2020)\notag
\end{align}
We assume a latent period of $\Delta=6$ days.\\
Our estimation results are given here:

{\centering 
\begin{tabular}{ |p{1.2cm}||p{0.6cm}|p{0.6cm}|p{0.6cm}|p{0.6cm}|p{0.8cm}||p{0.7cm}|}
 \hline
 Phases & $q$ & $\beta$ & $\sigma$ & $\gamma$ & $\kappa$ & $R_{0}$\\
 \hline
 Phase 2 & 0.7 & 0.06 & 0.22 & 0.15 & -0.10 & 1.004\\
 \hline
 Phase 3 & 0.6 & 0.07 & 0.15 & 0.15 & -0.05 &1.156\\
 \hline
 Phase 4 & 0.6 & 0.07 & 0.08 & 0.11 & 0.02 & 1.104\\
 \hline
\end{tabular}\par}
We note that, as the epidemic progresses, the basic reproduction number $R_{0}$ first rises, and then decreases with the implementation of consistent quarantine and other social distancing measures.\par

The prelimimary results also expose problems with real data based estimation and analysis. For example, due to the reduced availability of tests and test-sites in the early stage of the epidemic, as well as a non-random sample, the testing population presented in the data is severely skewed toward Symptomatic-Infected individuals. This hinders us from accurately capturing the true proportion of the Asymptomatic-Infected subgroup, as well as an accurate prevalence rate of infection over the total population.

In addition, our assumption of a constant latent period is not consistent with the nature of COVID-19; the latent period value we have used is an average value \cite{StephenA},\cite{Oran}. These issues lead to estimation errors, including the negative recovery rate values for $\kappa$ in Phase 1 and Phase 2. 

\section{Conclusions and Future Work} \label{sec:Future}
In this paper, we have briefly reviewed classical epidemiological compartment models, with a focus on a new SAIR(S) model that emphasizes the role played by the Asymptomatic-infected subpopulation. We presented continuous-time, discrete-time, and networked versions of the SAIR(S) model, and discuss their corresponding equilibria and some stability properties. We have noted the use of Nesterov's Next-Day Law and a basic least-squares approach for model parameter estimation, and conducted initial parameter estimation for COVID-19 using publicly available data from Champaign County, Illinois. Furthermore, we completed simulations of both group and networked models, investigated the impact of isolating subpopulations, and highlighted the crucial role of the Asymptomatic subgroup in the control of epidemic evolution. \par

In the estimation process, we have encountered many challenges, most significantly biased testing data and the lack of explicit information on the asymptomatic infected population. 
Our ongoing efforts include pursuing more complete stability and equilibria analyses for the N-SAIR(S) model and investigating approaches for model estimation under non-random and missing sample data sets, for example as described in \cite{CopasLi97}. We are further investigating Bayesian statistical methods for estimating true prevalence of epidemics under biased information on apparent prevalence.

\section*{Acknowledgments}
The authors would like to thank Dr. Joseph Kim, M.D., Ph.D., for many useful and interesting discussions, and for providing informative references. 
This work has been supported in part by Jump-ARCHES endowment through the Health Center for Engineering Systems Center at the University of Illinois, Urbana-Champaign, the C3.ai Digital Transformation Institute, and the NSF grant ECCS-2032321.

\bibliographystyle{IEEEtran}
\bibliography{SAIRS}


\end{document}